\documentclass[journal]{IEEEtran}
\usepackage[utf8x]{inputenc} 
\usepackage{amsfonts}
\usepackage{amsmath}
\usepackage{amssymb}
\usepackage{enumerate}
\usepackage{flafter}
\usepackage{fleqn}
\usepackage{latexsym}
\usepackage{mathrsfs}
\usepackage{psfrag}
\usepackage{graphicx}
\usepackage{subfigure}
\usepackage{url}
 \usepackage{color}
\usepackage{times}
\usepackage{flushend}
\usepackage{bm}

\usepackage{multirow}
\usepackage{multicol}
\usepackage{arydshln} 
\usepackage{diagbox}
\usepackage[colorlinks,citecolor=blue]{hyperref}

\newtheorem{myass}{Assumption}
\newtheorem{mydef}{Definition}
\newtheorem{mylem}{Lemma}
\newtheorem{myrem}{Remark}
\newtheorem{mythm}{Theorem}

\renewcommand{\thesection}{\arabic{section}.}

\date{ }
\makeatother

\long\def\begincomment#1\endcomment{}

\ifCLASSINFOpdf
\else
\fi

\begin{document}
%
\title{Practical Fractional-Order Variable-Gain Super-Twisting Control with Application to Wafer Stages of Photolithography Systems}
%
%
%

\author{Zhian~Kuang,~\IEEEmembership{Student Member,}
        Liting~Sun,
        Huijun~Gao,~\IEEEmembership{Fellow,~IEEE,}
        Masayoshi~Tomizuka,~\IEEEmembership{Life~Fellow,~IEEE}
\thanks{Zhian Kuang and Huijun Gao are with the Research Institute of Intelligent Control and Systems, Harbin Institute of Technology, Harbin,150001, China. Huijun Gao is also with the State Key Laboratory of Robotics and System, Harbin Institute of  Technology, Harbin 150001, China.
}
\thanks{Zhian Kuang, Liting Sun and Masayoshi Tomizuka are with the Mechanical Control System Lab, Mechanical Engineering Department, University of California, Berkeley, CA 94720, USA.

Corresponding Author: Huijun Gao(e-mail: hjgao@hit.edu.cn).
}
}

\markboth{Journal of \LaTeX\ Class Files,~Vol.~14, No.~8, August~2015}%
{Shell \MakeLowercase{\textit{et al.}}: Bare Demo of IEEEtran.cls for IEEE Journals}




\maketitle

\begin{abstract}
In this paper, a practical fractional-order variable-gain super-twisting algorithm (PFVSTA) is proposed to improve the tracking performance of wafer stages for semiconductor manufacturing. Based on the sliding mode control (SMC), the proposed PFVSTA enhances the tracking performance from three aspects: 1) alleviating the chattering phenomenon via super-twisting algorithm and a novel fractional-order sliding surface~(FSS) design, 2) improving the dynamics of states on the sliding surface with fast response and small overshoots via the designed novel FSS and 3) compensating for disturbances via variable-gain control law. Based on practical conditions, this paper analyzes the stability of the controller and illustrates the theoretical principle to compensate for the uncertainties caused by accelerations. Moreover, numerical simulations prove the effectiveness of the proposed sliding surface and control scheme, and they are in agreement with the theoretical analysis. Finally, practice-based comparative experiments are conducted. The results show that the proposed PFVSTA can achieve much better tracking performance than the conventional methods from various perspectives.
\end{abstract}


%
\IEEEpeerreviewmaketitle

\section{Introduction}
%
%
%
%
\IEEEPARstart{S}{emiconductor} manufacturing involves many precision devices, one of which is the wafer scanner~\cite{li2019kalman}. A wafer scanner is an optomechanical device used to finish the photolithography task, the principle of which is illustrated in Fig.~\ref{fig:wafer_scanner}~\cite{mishra2007precision}. The reticle, which contains the integrated circuit (IC) patterns on it, is placed onto the reticle stage, and the wafer stage carries the wafer. A fixed laser beam is generated by the illumination system. As the reticle stage and the wafer stage move, the reticle's IC patterns will be scanned and projected onto the wafer via lenses. To guarantee high yielding rate and quality, manufacturers require that both the reticle stage and the wafer stage track the designed trajectories in a fast and precise manner, i.e., the overlay error should be within a few nanometers~\cite{stearns2011iterative}.

\begin{figure}[http]
  \centering
   \includegraphics[width=140pt]{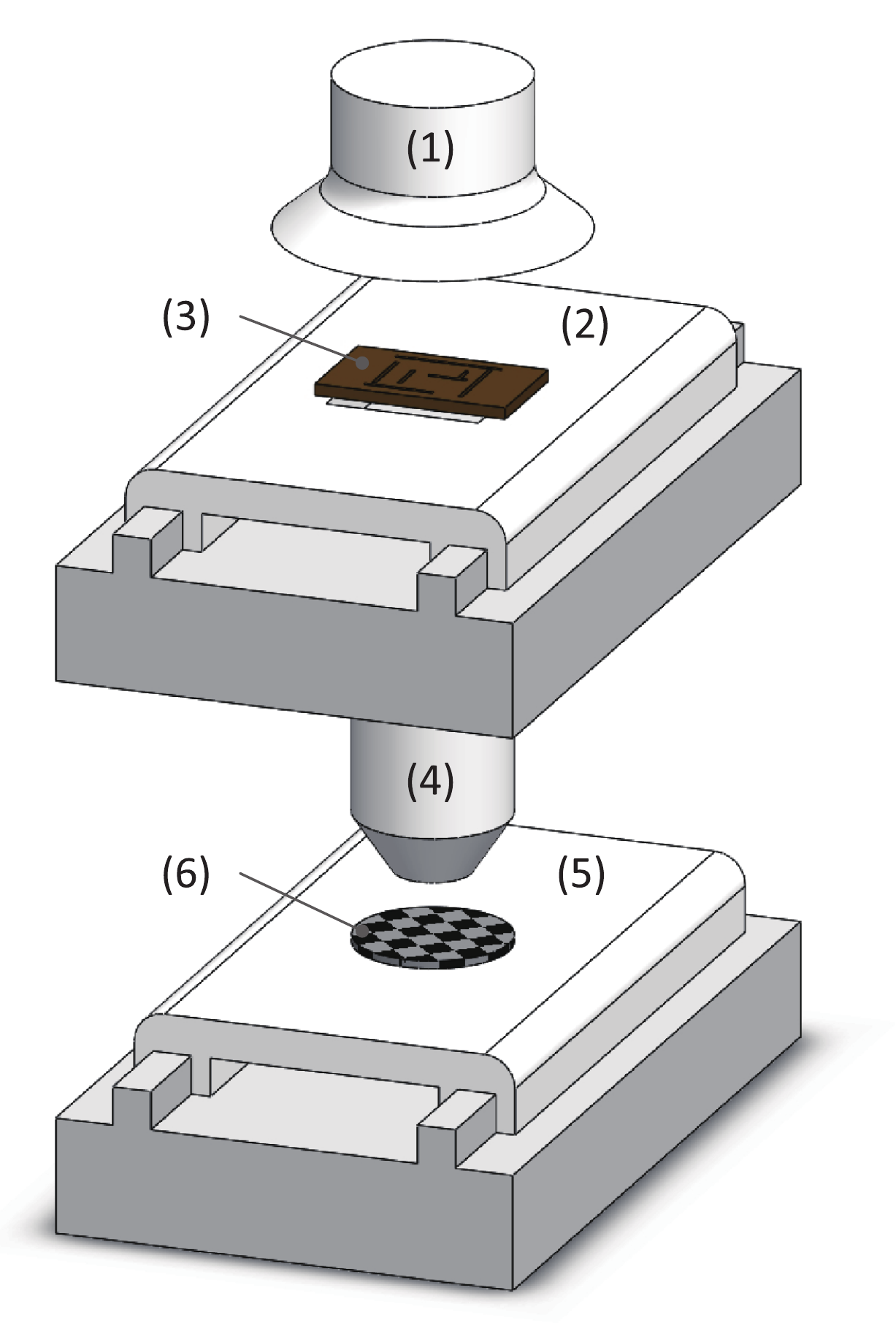}
   \caption{Schematic illustration of a photolithography system, which includes (1)~an illumination system, (2)~a reticle scanning stage, (3)~a reticle, (4)~projection lens, (5)~a wafer stage, (6)~a wafer.}
   \label{fig:wafer_scanner}
   \end{figure}


Practically, it is not trivial to achieve such high-precise tracking performance, because of the various disturbances and the inherent movements of wafer stages. For example, disturbances like external vibrations, force ripples, cable compliance, and structural vibrations can all occur in the wafer stage system~\cite{wang2015modified}. Some of these disturbances are even state-dependent, such as those caused by cables, which makes the application environment more complex~\cite{li2016state}. Moreover, the wafer stage's dynamic behavior is also position-dependent \cite{li2016state,wassink2005lpv}. As the movement of wafer stages contains both rapid acceleration phases and steady precise scanning phases, how to guarantee the performance for both phases, as well as how to balance the disturbance rejection and measurement noise sensitivity, is also a tricky problem to be solved~\cite{heertjes2012self}.

In the practical control of wafer stages, engineers usually apply both feedforward control and feedback control method to obtain excellent performance. Some researchers focus on the improvement of feedforward control techniques, such as the widely used iterative learning control (ILC) \cite{boeren2016frequency, zheng2017design, mishra2009projection, zhu2019internal}. The ILC method requires no plant information and is easy to be implemented, but the performance will be degraded when non-repetitive disturbances exist, such as cable forces and structure vibrations~\cite{wang2015modified}. The other researchers focus on the feedback control methods, of which the most extensively utilized one is PID control~\cite{butler2011position}. However, the PID controller and other linear control methods like $H \infty$ control can not handle the "water-bed effect" properly~\cite{heertjes2012self, hunnekens2014synthesis}. Therefore, researchers proposed the variable-gain PID controller to fix this problem~\cite{lee2009two,li2015data}. Nevertheless, due to the position-dependent dynamics and disturbances in the wafer stage, the actual performance achieved by these PID control methods is limited~\cite{wang2015modified,li2016state}. 
Recently, the sliding mode control (SMC) is implemented to wafer stages by researchers inspired by its effectiveness in improving the tracking performance in the presence of uncertainties and disturbances \cite{wang2015modified, li2016state, utkin2009sliding, kuang2019precise}. Results have shown that SMC can guarantee the robustness to both position-dependent disturbances and non-repetitive disturbances in the wafer stage system~\cite{wang2015modified, li2016state}.
However, traditional SMC implementations suffer from the well-known chattering problem due to its variable structure~\cite{du2016chattering, kuang2018contouring}
. This shortcoming might significantly deteriorate the achievable performance of the SMC controller~\cite{du2016chattering}.

To solve the chattering problem, Levant proposed the super-twisting algorithm (STA) in~\cite{levant1993sliding}, and STA turns out to be one of the most effective method~\cite{moreno2012strict, evangelista2012lyapunov}. In the traditional STA, the gains in the switching control law are constants, denoted as constant-gain super twisting algorithm~(CGSTA) in this paper. This kind of controller can reduce the chattering phenomenon but can only deal with the time-invariant and state-independent disturbances and uncertainties~\cite{gonzalez2012variable,bera2015variable}, which also suffers from the water-bed effect. Based on STA, a variable-gain super-twisting algorithm (VGSTA) was proposed in~\cite{gonzalez2012variable} to further enhance the performance in terms of introducing variable parameters in STA. This method provides sufficient compensation for uncertainties and disturbances, which can promote the robustness of the system. However, the boundary functions of the uncertainties/disturbances in this method are hard to know, which brings obstacles to design the controller in the practical applications on wafer stages.

Moreover, a common drawback of the two kinds of STA above is the slow responses due to the use of linear sliding surfaces. 
To improve the performance further, researchers have proposed some advanced sliding surfaces to develop STA further. For instance, 
Based on the integral sliding surface~(ISS), the integral super-twisting algorithm~(ISTA) is derived in~\cite{kurkccu2018disturbance}. This control algorithm guarantees the fast response of the states on the sliding surface but has a rather significant overshoot. The fractional-order sliding surface~(FSS) has emerged as a useful control scheme towards this issue~\cite{kuang2020precise, monje2010fractional, efe2010fractional, ma2020fuzzy}. As a generalization of the integer-order sliding surface, it can simultaneously guarantee fast response and small overshoot of the dynamics~\cite{kuang2018, sun2018practical}. Some researches also show that FSS can help to reduce the chattering phenomenon further~\cite{zhang2012fractional}. Due to these merits, it has been used in various mechatronic systems with excellent performance~\cite{ISI:000464942600017, wang2016practical, kuang2018}, and a new fractional-order constant-gain super-twisting algorithm~(FCGSTA) is developed based on this~\cite{wang2019adaptive}.  However, as this controller only has constant gains, it cannot thoroughly compensate for the uncertainties and disturbances. Moreover, the conventional design of the FCGSTA consists of too many hyper-parameters, which are hard to tune in the application to wafer stages. 

In ultra-precision motion control tasks on wafer stages, the controller needs to satisfy: 
\begin{itemize}
    \item fast response to track the reference signal;
    \item robust performance in the presence of disturbances and uncertainties;
    \item high performance in different phases;
    \item practically easy to design.
\end{itemize}
Motivated by such goals, a practical fractional-order variable-gain super-twisting algorithm (PFVSTA) is proposed in this paper. The proposed method's sliding surface is designed with the term of fractional-order (FO) calculus, which can guarantee the states reach the equilibrium points with fast response and small overshoot. Moreover, variable gains are introduced in the controller so that the uncertainties are handled properly. Specifically, the wafer stage's acceleration is considered, and the controller can get better performance than traditional methods. Finally, the controller is developed based on the practical model and is easy to implement. With all these merits, different disturbances and uncertainties, as well as the water-bed effect in the wafer stage system, can be handled appropriately by the proposed method. 

This paper's contribution is four-fold: 1) A specially designed fractional-order sliding surface is given out. By introducing a particular term, the convergence of the stated on the sliding surface is accelerated, making the controller have better performance without increasing tunable parameters. 2) A novel variable-gain switching controller is designed and analyzed theoretically and practically, which guarantees the robustness of the controller under model uncertainties and disturbances. By designing a novel form of the super-twisting algorithm, the sliding variable dynamics is further improved. 3) In the designed controller, the influence of the feedforward term is analyzed for the wafer stage system. Moreover, with acceleration-based variable gains, the acceleration's influence is appropriately handled, so that the precision of the wafer stage system is further improved. 4) The proposed controller is applied and verified in simulations and experiments on a real wafer scanning stage testbed.

The remainder of this paper is organized as follows. Section~2 presents the model of the wafer stage and the designed control strategy.
Section~3 presents a stability analysis of the proposed PFVSTA.
Simulation and experimental results are given in Section~4. Finally, Section~5 concludes the paper.

\section{Controller Design} 
\label{sec: controller design}
\subsection{Model of Wafer Stage}\label{subsec: system_model}
\begin{figure}[http]
  \centering
   \includegraphics[width=230pt]{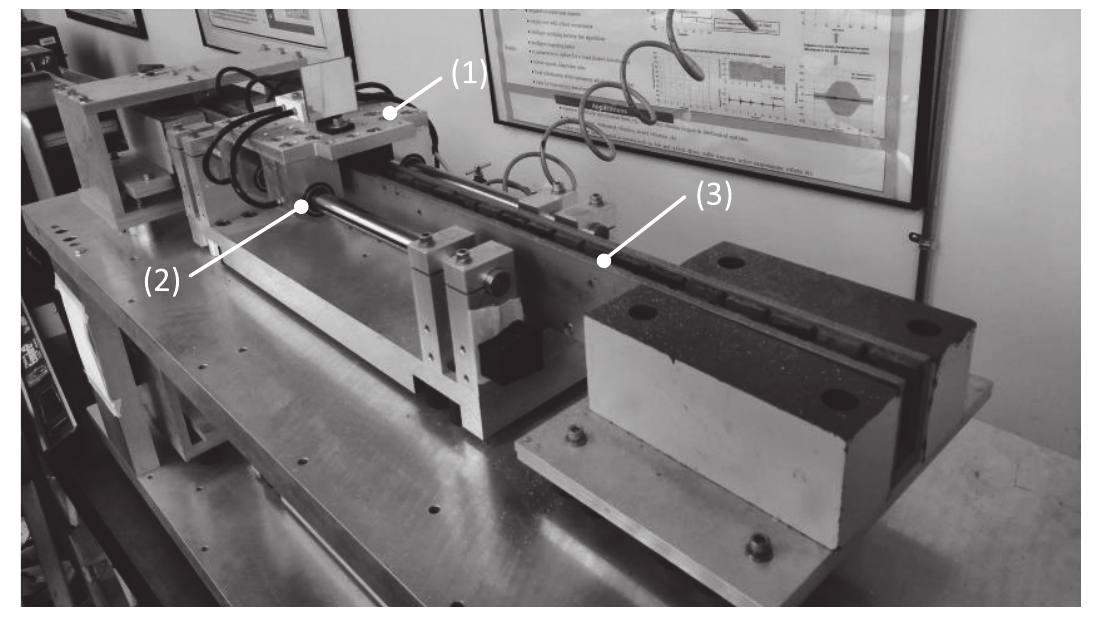}
   \caption{The wafer stage testbed including three parts: (1)~the moving stage, (2)~the air bearings, and (3)~the linear motor.}
   \label{fig:wafer_stage}
   \end{figure}

As shown in Fig.~\ref{fig:wafer_stage}, the testbed in this paper is a one-dimensional stage. It is supported via air bearings, and directly driven by a linear motor along with the sliding guide. Based on Newton's second law,  the dynamics of the wafer stage can be written as \cite{kuang2019simplified}
\begin{align} \label{eq:Newton}
\dot{v}=(F-f)/m
\end{align}
where $v$ is the velocity of the stage, $\dot{v}$ is the acceleration of the stage, $f$ denotes the sliding friction, $m$ is the mass of the stage, and $F$ is the actuation force exerted by the linear motor. 

Benefiting from the air bearings, the sliding friction can be approximated by
\begin{align} \label{eq:friction}
 f=K_v v+d_f
\end{align}
where $K_v$ is the viscous friction coefficient, and $d_f$ is the equivalent un-modelled part of the friction, i.e., friction uncertainties.

The actuation force $F$ is proportional to the input voltage $u$ and it can be written as
\begin{align} \label{eq:driven_force}
 F=Q u + d_r
\end{align}
where $Q [N/V]$ is the ratio between the actuation force and input voltage, and $d_r$ stands for unknown disturbances caused by, for instance, the force ripple of the motor~\cite{tan2002robust}.

Substituting (\ref{eq:friction}) and (\ref{eq:driven_force}) into (\ref{eq:Newton}), the system dynamics become
\begin{align} \label{model:natural model}
  \dot{v}=-T_v v+K u+d
\end{align}
with $T_v=\frac{K_v}{m}$, $K=\frac{Q}{m}$ and $d=\frac{d_f+d_r}{m}$.

Note that in (\ref{model:natural model}), $d$ is unknown and influences the system dynamics as a lumped disturbance. Hence, we can obtain the following nominal model: 
\begin{align}
 \dot{v}=-\bar{T}_{v} v+\bar{K}u
\end{align}
where $\bar{T}_{v}$ and $\bar{K}$ are the nominal values of ${T}_v$ and ${K}$ respectively. The uncertainties are represented in an additive form as
\begin{align}
 K&=~\bar{K}+\Delta K\\
 T_v&=~\bar{T}_{v}+\Delta T_v
\end{align}
where $\Delta K$ and $\Delta T_v$ stand for the uncertain parts in $K$ and $T_v$ respectively. 

\subsection{PFVSTA}

 \begin{figure}[http]
    \centering
     \includegraphics[width=230pt]{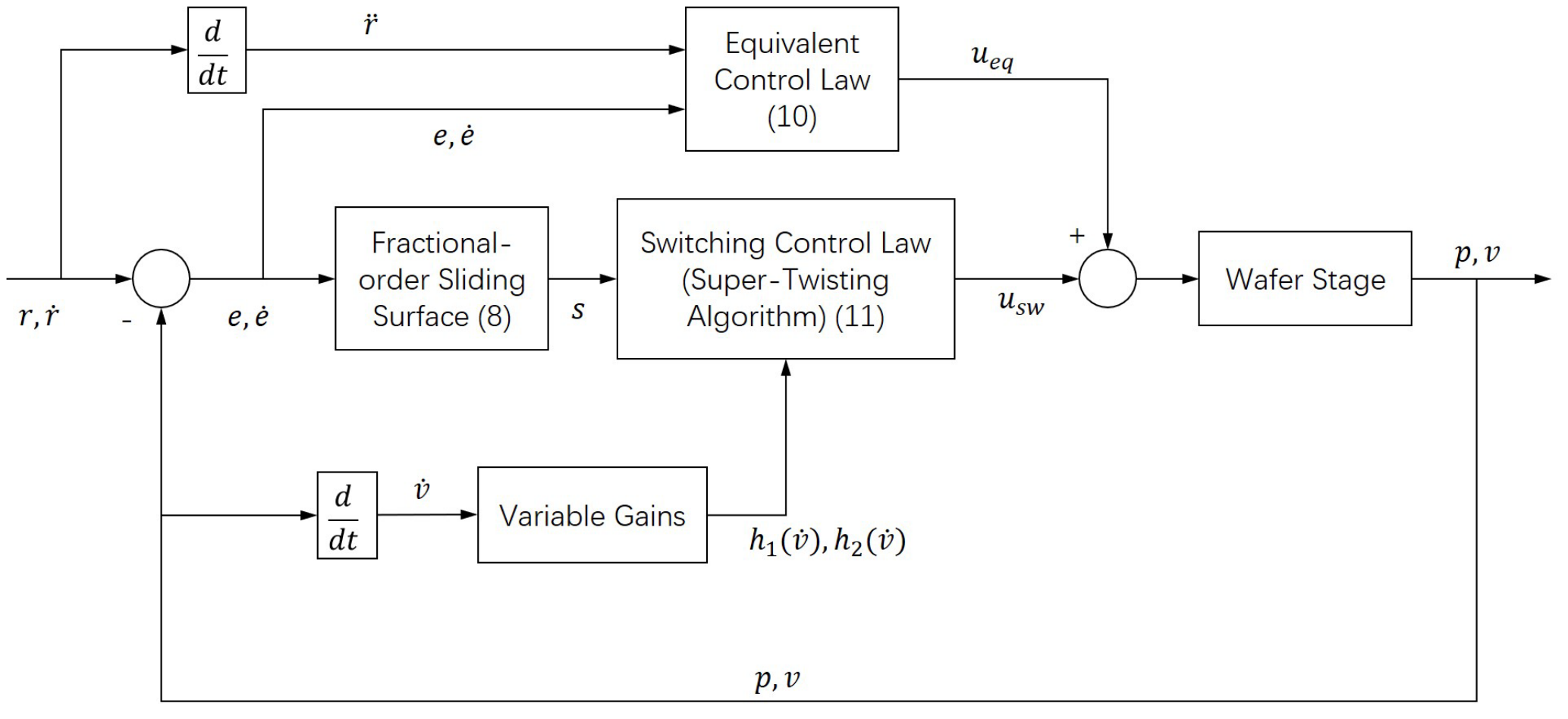}
     \caption{Block diagram of the proposed controller}

     \label{fig:structure}
     \end{figure}

For the succeeding PFVSTA design, the following assumption is imposed:
\begin{myass} \label{assumption:1}
  The parametric uncertainties of the model and the disturbance in the system are bounded, i.e., $\Delta K $, $\Delta T_v$ and $d$ are bounded.
\end{myass}

\label{sec:control_design}
In this paper, the tracking error is defined as~$
e=p-r$,
where $e$ is the tracking error, $p$ is the actual position of the stage, $r$ stands for the reference position, and it is apparent that $\dot{p}=v$. The sliding mode controller will be designed on the basis of $e$. Here, we have another assumption:
\begin{myass} \label{assumption2}
 The actual velocity $v$ and the reference velocity are Lipschitz functions, i.e.,  $\dot{v}$ and $\ddot{r}$ are bounded.
\end{myass}

A novel fractional-order sliding surface is presented as 
\begin{align} \label{def:sliding surface}
s=\dot{e}+k_1D^{\xi-1}\left[\textrm{sig}(e)^a\right]+k_2\textrm{sig}(e)^\frac{1}{a}, \xi,a \in (0,1)
\end{align}
where $k_1$ and $k_2$ are positive constants, $\textrm{sig}(\bullet)^*=\textrm{sgn}(\bullet)|\bullet|^{*}$, $D^{*}$ is fractional-order calculus operator and its definition is stated as Definition~\ref{def: fractional-order calculus} in Appendix A. 

As shown in Fig.~\ref{fig:structure}, we design the sliding mode controller on the equivalent-control basis~\cite{cucuzzella2018practical} and it is written as
\begin{align} \label{eq:u}
u=u_{eq}+u_{sw}
\end{align}
where $u_{eq}$ is the equivalent controller and $u_{sw}$ is the switching controller with variable-gain super-twisting algorithm.

According to the theory of SMC, the switching control law is implemented to make the sliding variable $s$ converge to 0, i.e., to make the states converge to the sliding surface. In this paper, a novel switching controller with VGSTA is addressed as
\begin{align} \label{eq:usw}
  u_{sw}=-\frac{h_1(\dot{v})}{\bar{K}}\Phi_1(s)-\int_0^{t}\frac{h_2(\dot{v})}{\bar{K}}\Phi_2(s)dt
\end{align}
with \begin{align}
    \Phi_1(s)&=|s|^{\alpha(s)
  }\textrm{sgn}(s) \label{phi1}\\
    \Phi_2(s)&=\Phi_1^{'}(s)\Phi_1(s) \label{phi2}\\
    \alpha(s)&=\frac{4|s|+1}{2\left(|s|+1\right)}
\end{align}
where $h_1(\dot{v})$ and $h_2(\dot{v})$ are the variable gains.

The equivalent control $u_{eq}$ is implemented to maintain the states on the sliding surface when no uncertainties and disturbances are considered. Therefore, by setting $\dot{s}=0$ and considering the model of the wafer stage, $u_{eq}$ can be obtained as
\begin{align}  \label{equivalent control}
u_{eq}=\frac{1}{\bar{K}}\left( \ddot{r}-k_1 \frac{dD^{\xi-1}\left[\textrm{sig}(e)^a\right]}{dt}-\frac{k_2}{a}|e|^{\frac{1-a}{a}}\dot{e}+\bar{T}_{v}v \right).
\end{align}



\begin{myrem}
For the application on wafer stages, the designed sliding surface~(\ref{def:sliding surface}) guarantees a fast and smooth response, which is beneficial for precision. The designed switching controller~(\ref{eq:usw}) uses a novel structure to make $s$ converge fast and steadily. The variable-gain structure also has an advantage over the traditional methods, especially when the reference acceleration is large.
From Fig.~\ref{fig:structure} and (\ref{equivalent control}), it can be seen that our $u_{eq}$ contains both feedback and feedforward terms. 
In the precision control of wafer stages, feedforward techniques are extensively used to enhance the performance~\cite{boeren2016frequency, zheng2017design}. The feedforward term in our method not only guarantees the completion of the sliding mode controller but also improves the performance of the wafer stage. All these merits of our proposed method will be verified by simulations and experiments in Section~{\ref{sec:sims and exps}}. 
\end{myrem}

\section{Dynamic Analysis of PFVSTA} \label{sec:dynamic analysis of FVGSTA Scheme}

The dynamics of $s$ can be obtained by substituting~(\ref{eq:u}) into the natural model~(\ref{model:natural model}),
\begin{align} \label{eq:s}
  \dot{s}&=-\frac{K}{\bar{K}}h_1(\dot{v})\Phi_1(s)-\frac{K}{\bar{K}}\int_0^{t}h_2(\dot{v})\Phi_2(s)dt \nonumber\\  
&~~~+\left(\frac{K}{\bar{K}}-1\right)\left(\ddot{r}-k_1\frac{dD^{\xi-1}\left[\textrm{sig}(e)^a\right]}{dt}-k_2\dot{e}\right)\nonumber \\  
&~~~+\left(\frac{K}{\bar{K}}\bar{T}_{v}-T_v\right)v+d.
\end{align}

Equation~(\ref{eq:s}) can be further rewritten as
\begin{align}\label{eq:dot_of_s}
\dot{s}&=-h_1(\dot{v})\Phi_1(s)-\int^t_0h_2(\dot{v})\Phi_2(s)dt +g_1(\dot{v}) 
\nonumber \\&~~~
+g_2(\dot{v})
\end{align}
with \begin{align}\label{eq:g1}
g_1(\dot{v})&=-\frac{\Delta K}{\bar{K}} h_1(\dot{v})\Phi_1(s)+\frac{\Delta K}{\bar{K}} \ddot{r}
\nonumber\\
&~~~
-\frac{\Delta K}{\bar{K}} k_1 D^\xi[\textrm{sig}(e)^a]+d
\end{align}
and\begin{align}\label{eq:g2}
g_2(\dot{v})=
\left(\frac{K}{\bar{K}}\bar{T}_v-T_v\right)v-\frac{\Delta K}{\bar{K}} \int^t_0h_2(\dot{v})\Phi_2(s).
\end{align}

Based on (\ref{eq:g1}), (\ref{eq:g2}), Assumption~\ref{assumption:1} and Assumption~\ref{assumption2}, we obtain that $|g_1(\dot{v})|$ and $\left|\frac{d g_2(v,\dot{v})}{dt}\right|$ are bounded with
\begin{align}
    |g_1(\dot{v})|
    & \leq D_1 |\Phi_1(s)|+D_2
    \\ \label{ineq:boundary of dg2}
    \left|\frac{d g_2(v,\dot{v})}{dt}\right|
    &\leq D_3|\dot{v}|+D_4|\Phi_2(s)|
\end{align} where $D_1$, $D_2$, $D_3$ and $D_4$ are positive constants that satisfy $D_1\geq \left|\frac{\Delta K}{\bar{K}}\right| |h_1(\dot{v})|$, $D_2 \geq \left|\frac{\Delta K}{\bar{K}}\right||\ddot{r}|+\left|k_1\frac{dD^{\xi-1}\left[\textrm{sig}(e)^a\right]}{dt}\right|+|d|$, $D_3\geq \left|\frac{K}{\bar{K}}\bar{T}_v-T_v\right|$ and $D_4\geq \left|\frac{\Delta K}{\bar{K}}\right||h_2(\dot{v})|$.

The following lemma is useful in the further analysis:
\begin{mylem} \label{lemma}
    (see \cite{bandyopadhyay2015stabilization, dadras2012fractional, kilbas2006theory}) For the fractional-order integration operator ${}_0I^{\alpha}_t$, the following relation holds:
    \begin{align}
        |{}_0I^{\alpha}_t f(t)|\leq \kappa|f(t)|
    \end{align}
    where $\kappa$ is a finite positive constant.
\end{mylem}

\begin{mythm} \label{thm: stability of s}
For the system~(\ref{model:natural model}), if the variable gains $h_1(\dot{v})$ and $h_2(\dot{v})$ of the controller (\ref{eq:u})-(\ref{eq:usw}) are selected as
\begin{align} \label{eq:h_1}
    h_1(\dot{v})&=\frac{p_4}{p_2^3-p_1p_2p_4}\left( \frac{(-p_2 \delta_1(\gamma)+p_4\delta_2(\dot{v}))^2}{4} \right.\nonumber \\
    &~~~\left. +\frac{p_2^2p_1}{p_4}+p_2^2\delta_2(\dot{v})-p_1p_2\delta_1(\gamma)\right)
\\
\label{eq:h_2}
    h_2(\dot{v})&=\frac{1}{p_4}\left(p_1-p_2h_1(\dot{v})\right)
\end{align}
with 
\begin{align} \label{ineq: the range of ps}
    p_1p_4-p_2^2&>0, p_1>0, p_2<0
\\ \label{delta1}
    \delta_1(\gamma)&=D_1+\frac{D_2}{\Phi_1(\gamma)} \\ \label{delta2}
    \delta_2(\dot{v})&=D_3 |\dot{v}|+D_4
\end{align}
then the sliding variable $s$ will converge to the region 
\begin{align}
    \Gamma = \left\{s \big| |s|\leq \gamma\right\}, \gamma>0.
\end{align}
\end{mythm}
\begin{proof}
For the sake of convenience, rewrite (\ref{eq:dot_of_s}) as
  \begin{align}\label{dynamics of s 2}
   \dot{s}&=-h_1(\dot{v})\Phi_1(s)+z+g_1(\dot{v}) \\
\dot{z}&=-h_2(\dot{v}) \Phi_2(s)+\frac{d g_2(\dot{v})}{dt}.
  \end{align}

  Construct a Lyapunov function candidate as
  \begin{align} \label{lyapunov function}
  V &= \bm{\Theta}^{\bm{T}} \bm{P} \bm{\Theta} \nonumber 
  \\
  &=
 \left[\begin{array}{cc}
      \Phi_1(s) &
      z
    \end{array}\right] 
    \left[ \begin{array}{cc}
      p_1 & p_2 \\
      p_2 & p_4
    \end{array}
    \right] 
    \left[
      \begin{array}{c}
        \Phi_1(s) \\ z
      \end{array}
    \right]
  \end{align}
where $\bm{\Theta}$ is defined as $
   \bm{\Theta}=
    \left[
    \begin{array}{cc}
      \Phi_1(s) & z
    \end{array}
    \right]^{\bm{T}}$
  and $\bm{P}$ is defined as
$
\bm{P}=
  \left[
    \begin{array}{cc}
    p_1 & p_2 \\
      p_2 & p_4
    \end{array}
  \right]$.
From~(\ref{ineq: the range of ps}), we find that $\bm{P}\geq 0$, which satisfies the condition of being selected as a Lyapunov function candidate.
  
The derivative of the Lyapunov function candidate is
  \begin{align}
\dot{V}
    =\left(\frac{d \bm{\Theta}}{dt}\right)^{\bm T} \bm{P \Theta}+\bm{\Theta^T P}\frac{d \bm{\Theta}}{dt}.
    \label{eq:derivative of V}
 \end{align}  
  
In~(\ref{eq:derivative of V}), $\frac{d \bm{\Theta}}{dt}$ is calculated as
\begin{align}
   \frac{d \bm{\Theta}}{dt}&=
\left[
\begin{array}{cc}
  \Phi_1^{\prime}(s)\dot{s} & \dot{z}
\end{array}
\right]^{\bm{T}} \nonumber \\
  &=\Phi_1^{'}(s) 
\left[
\begin{array}{cc}
l_1(\dot{v}) & 1\\
l_2(\dot{v}) & 0
\end{array}
\right]\left[
  \begin{array}{c}
    \Phi_1(s) \\
    z
  \end{array}
\right] \label{eq:derative of Theta}
\end{align} 
where
$l_1(\dot{v})$ and $l_2(\dot{v})$ are defined as
\begin{align}
     l_1(\dot{v})&=-h_1(\dot{v})+\frac{g_1(\dot{v})}{\Phi_1(s)} \\
    l_2(\dot{v})&=-h_2(\dot{v})+\frac{dg_2(v)}{dt}\frac{1}{\Phi_2(s)}
\end{align}
respectively for the sake of simplicity.  

By substituting~(\ref{eq:derative of Theta}) into~(\ref{eq:derivative of V}), we have
\begin{align}
    \dot{V}&=\Phi_1^{'}(s) \bm{\Theta}^{\bm{T}}
    \left[
    \begin{array}{cc}
      2p_1l_1(\dot{v})+2p_2l_2(\dot{v})   & * \\
       p_1+p_2l_1(\dot{v})+p_4 l_2(\dot{v})  & 2p_2
    \end{array}
    \right]\bm{\Theta} \nonumber \\
    &\triangleq 
    \Phi_1^{'}(s) \bm{\Theta}^{\bm{T}}
    \left[
    \begin{array}{cc}
     A  & B \\
     B  & C
    \end{array}
    \right]\bm{\Theta}.
\end{align}

Considering the definition of $l_1(\dot{v})$ and $l_2(\dot{v})$, we have\begin{align} \label{eq:A}
    A&=-2p_1h_1(\dot{v})-2p_2h_2(\dot{v})+2p_1\frac{g_1(\dot{v})}{\Phi_1(s)} 
    \nonumber \\&~~~
    +2p_2 \frac{dg_2(v)}{dt \Phi_2(s)}.
\end{align}

Substitute~(\ref{eq:h_1}) and~(\ref{eq:h_2}) into~(\ref{eq:A}) to obtain
\begin{align} \label{eq:A2}
    A&=\left(\frac{2p_2^2}{p_4}-2p_1 \right)h_1(\dot{v})-\frac{2p_1p_2}{p_4}+2p_1\frac{g_1(\dot{v})}{\Phi_1(s)}
    \nonumber \\
    &~~~+2p_2\frac{dg_2(\dot{v})}{dt}\frac{1}{\Phi_2(s)}
    \nonumber \\
    &=\frac{(-p_2\delta_1(\gamma)+p_4\delta_2(\dot{v}))^2}{2p_2}+2p_1\left(\frac{g_1(\dot{v})}{\Phi_1(s)}-\delta_1(\gamma)\right) \nonumber \\
    &~~~+2p_2\left(\delta_2(\dot{v})+\frac{dg_2(\dot{v})}{dt}\frac{1}{\Phi_2(s)}\right).
\end{align}

When the sliding variable $s$ is not in the region $\Gamma$, i.e., $|s|>\gamma>0$, we get that $\Phi_1(\gamma)<|\Phi_1(s)|$.
Hence, according to (\ref{delta1}) and (\ref{eq:g1}),
\begin{align}\label{ineq:delta1}
\delta_1(\gamma) =  D_1+\frac{D_2}{\Phi_1(\gamma)}
>
D_1+\frac{D_2}{|\Phi_1(s)|}
\geq
\frac{|g_1(\dot{v})|}{|\Phi_1(s)|}. 
\end{align}
For $\delta_2(\dot{v})$, we can get that
\begin{align} \label{ineq:delta2}
   \delta_2(\dot{v})=D_3 \frac{\dot{v}}{|\Phi_2(s)|}+D_4
   \geq
   \frac{1}{|\Phi_2(s)|}\left|\frac{dg_2(\dot{v})}{dt}\right|.
\end{align}

Considering that $p_1>0$, $p_2<0$ and substituting (\ref{ineq:delta1}) and (\ref{ineq:delta2}) into (\ref{eq:A}), finally we attain
$
    A<0.
$

On the other hand, 
\begin{align} \label{eq:tem}
\left|
\begin{array}{cc}
    A & B \\
    B & C
\end{array}
\right|
&=2p_2(2p_1l_1(\dot{v}))+4p_2^2l_2(\dot{v})\nonumber \\
&~~~-\left(p_2\frac{g_1(\dot{v})}{\Phi_1(s)}+p_4\frac{dg_2(v)}{dt\Phi_2(s)}\right)^2 \nonumber  \\
&=-4p_1p_2h_1(\dot{v})-4p_2^2h_2(\dot{v})\nonumber \\
&~~~+4p_1p_2\frac{g_1(\dot{v})}{\Phi_1(s)}+4p_2^2\frac{dg_2(v)}{dt\Phi_2(s)}
\nonumber \\&
~~~-\left(p_2\frac{g_1(\dot{v})}{\Phi_1(s)}+p_4\frac{dg_2(v)}{dt\Phi_2(s)}\right)^2.
\end{align}

By substituting~(\ref{eq:h_1}) and~(\ref{eq:h_2}) to (\ref{eq:tem}), we obtain
\begin{align}
    \left|
\begin{array}{cc}
    A & B \\
    B & C
\end{array}
\right|
&= 4p_2^2 \left( \frac{dg_2(\dot{v})}{dt\Phi_2(s)}+\delta_2(\dot{v})\right)-\left(p_2\frac{g_1(\dot{v})}{\Phi_1(s)}\right. \nonumber \\
&~~~\left.+p_4\frac{dg_2(v)}{dt\Phi_2(s)}\right)^2+4p_1p_2\left(\frac{g_1(\dot{v})}{\Phi_1(s)}-\delta_1(\gamma)\right) \nonumber \\
&~~~+ \left(-p_2\delta_1(\gamma)+p_4\delta_2(\dot{v})\right)^2
\nonumber 
>0.
\end{align}
Then based on $\Phi_{1}^{'}(s)> 0$ when $|s|>\gamma$, it is obvious that
\begin{align}
  \Phi_1^{'}(s)\bm{\Theta}^{\bm{T}}
\left[
\begin{array}{cc}
A & B \\
B & C
\end{array}
\right]
\bm{\Theta}< 0.
\end{align}

It has been shown that $V$ is indeed a Lyapunov function. This completes the proof of Theorem~\ref{thm: stability of s}.
\end{proof}

\begin{myrem} \label{remark1}
Theorem~\ref{thm: stability of s} not only states the sufficient conditions for the stability, but also gives the relation between the variable gains $h_1(\dot{v})$, $h_2(\dot{v})$ and the region $\Gamma$. Based on this relation, we can analyze how the acceleration affect the performance as follows:

Consider there are two states that $\dot{v}_1$ and $\dot{v}_2$, suppose that $|\dot{v}_2|>|\dot{v}_1|$ and let $p_1$, $p_2$ and $p_4$ remain the same to make the comparison fair, then we have 
\begin{align}
    h_1(\dot{v}_1)&=\frac{p_4}{p_2^3-p_1p_2p_4}\left( \frac{(p_2 \delta_1(\gamma_1)+p_4\delta_2(\dot{v}_1))^2}{4} \right.\nonumber \\
    &~~~\left. +\frac{p_2^2p_1}{p_4}+p_2^2\delta_2(\dot{v}_1)-p_1p_2\delta_1(\gamma_1)\right)
\\
    h_1(\dot{v}_2)&=\frac{p_4}{p_2^3-p_1p_2p_4}\left( \frac{(p_2 \delta_1(\gamma_2)+p_4\delta_2(\dot{v}_2))^2}{4} \right.\nonumber \\
    &~~~\left. +\frac{p_2^2p_1}{p_4}+p_2^2\delta_2(\dot{v}_2)-p_1p_2\delta_1(\gamma_2)\right)
\end{align}
and $\delta_2(\dot{v}_2)>\delta_2(\dot{v}_1)$ from~(\ref{delta2}). In constant-gain super-twisting algorithms, $h_1(\dot{v}_1)$ and $h_2(\dot{v}_2)$ have the relation that $h_1(\dot{v}_1)=h_1(\dot{v}_2)=C$.
Thus we have $\delta_1(\gamma_2)<\delta_1(\gamma_1)$ and $\gamma_2>\gamma_1$, which means that a larger absolute value of acceleration brings a larger $\gamma$ under the constant-gain super-twisting control. Based on this, we can infer that large acceleration has a negative influence on the precision of the system when the constant-gain super-twisting controller is applied. 

As for the variable-gain super-twisting controller, $h_1(\dot{v})$ can be designed as an increasing function, i.e., $h_1(\dot{v}_1)<h_2(\dot{v}_2)$. Ideally, if $h_1(\dot{v}_1)$ and $h_1(\dot{v}_2)$ are designed to satisfy that \begin{align} \label{eq:ideal}
    h_1(\dot{v}_2)-h_1(\dot{v}_1)&=\frac{p_4}{p_2^3-p_1p_2p_4}\left( \frac{p_4^2}{4}\left(\delta_2^2(\dot{v}_2)-\delta_2^2(\dot{v}_1)\right) \right.
    \nonumber \\
   &~~~+p_2^2(\delta_2(\dot{v}_2)-\delta_2(\dot{v}_1)) \nonumber \\
     &~~~\left.+\frac{p_2p_4\delta_1(\gamma_1)}{2}(\delta_2(\dot{v}_2)-\delta_2(\dot{v}_1))\right)
\end{align}
then we have $\gamma_2=\gamma_1$, which means that the negative influence of the acceleration is eliminated completely. In practice, due to the acceleration is usually time-varying and the variable gains are usually designed with a simple function,  (\ref{eq:ideal}) cannot always be satisfied. However, as long as $h_1(\dot{v})$ is designed properly, it is helpful to reduce the influence that the acceleration poses on the precision.   
\end{myrem}

\begin{mythm} \label{thm:tracking error}
  If the sliding variable $s$ reaches the region $\Gamma$, then the tracking error will converge to the region $E$, and
  \begin{align}
      E=\left\{ e(t)\big||e(t)|\leq \varepsilon\right\}
  \end{align}
  where $\varepsilon$ is the positive solution of the equation
  \begin{align} \label{eq:positive solution of the function}
      f(x)=k_1\kappa x^a-k_2 x^\frac{1}{a}+\gamma=0.
  \end{align}
\end{mythm}
\begin{proof}
  Construct another Lyapunov function candidate
  \begin{align}
    W(t)=|e(t)|.
  \end{align}
  The derivative of $W(t)$ is 
  \begin{align}
    \dot{W}(t)=\textrm{sgn}\left(e(t)\right)\dot{e}(t).
  \end{align}

From the definition of the sliding surface~(\ref{def:sliding surface}),
  \begin{align}
    \dot{e}(t)=-k_1D^{\xi-1}\left(\textrm{sig}^a(e(t))\right)-k_2\textrm{sig}^{\frac{1}{a}}(e(t))+s.
  \end{align}
  Then $\dot{W}(t)$ can be rewritten as
  \begin{align}
    \dot{W}(t)&=\textrm{sgn}(e(t))\left(-k_1D^{\xi-1}\left(\textrm{sig}^a(e(t))\right)-k_2e^{\frac{1}{a}}(t)+s\right) \nonumber\\
    &\leq|k_1 {}_{t_0}I^{1-\xi}|e(t)|^a \textrm{sgn}(e(t))|
 -k_2|e(t)|^{\frac{1}{a}}+sgn(e(t))s.\nonumber
 \end{align}
 From Lemma~\ref{lemma} and the precondition that $|s|\leq \gamma $, we can obtain
$
 \dot{W}(t)
 \leq k_1 \kappa|e(t)|^a-k_2|e(t)|^{\frac{1}{a}}+\gamma
$.
  If the positive solution of $f(x)=k_1\kappa x^a-k_2 x^{\frac{1}{a}}+\gamma=0$ is denoted as $\varepsilon$, then it is obvious that when $x>\varepsilon$,
     $f(x)> 0$.

  Thus, when $|e(t)|>\varepsilon$, we have
  \begin{align}
      \dot{W}(t)<0.
  \end{align}
  
This completes the proof.
\end{proof}

\begin{myrem}
According to Theorem~\ref{thm:tracking error}, the boundary of the stable state is the solution of~(\ref{eq:positive solution of the function}). Based on the property of function $f(x)$, a smaller $\gamma$ leads to a smaller $\varepsilon$ and further, higher precision of the system theoretically, which keeps consistent with our cognition. 
\end{myrem}

\begin{myrem}
When $0<a<1$, the positive solution of $f(x)=0$ always exist. We can assume that there are two functions $f_1(x)=k_1\kappa x^a+\gamma$ and $f_2(x)=k_2 x^{\frac{1}{a}}$. It is apparent that $f_1(0)>f_2(0)$. Because $f_1^{''}(x)=k_1\kappa a (a-1) x^{a-2}<0$ and $f_2^{''}(x)=\frac{k_2}{a}\left(\frac{1}{a}-1\right)x^{\frac{1}{a}-2}>0$, there always be $f_1(\varepsilon)=f_2(\varepsilon)$, i.e., $f(\varepsilon)=0$. 

For $a=1$, the sliding surface becomes $s=\dot{e}+k_1D^{\xi-1}[\textrm{sig}(e)]+k_2e$, which is a traditional fractional-order sliding surface. Through a similar analysis procedure, we can get that $f_1^{''}(x)=f_2^{''}(x)=0$, which means that if $k_1$, $k_2$ or $\kappa$ is selected improperly, a finite $\varepsilon$ cannot be obtained, i.e., the stability on the sliding surface cannot be guaranteed. 
\end{myrem}

\section{Simulations and Experiments} \label{sec:sims and exps}
In this part, simulations and experiments are conducted to verify the effectiveness of the proposed PFVSTA. To make the results convincing, six different controllers will be mentioned and compared in different aspects. They are listed here to facilitate understanding:

\begin{itemize}
    \item the traditional CGSTA~\cite{evangelista2012lyapunov}, $u=\frac{1}{\bar{K}}\left( \ddot{r}-k_2\dot{e}+\bar{T}_{v}v \right)-\frac{h_1}{\bar{K}}\Phi_1(s)-\int_0^{t}\frac{h_2}{\bar{K}}\Phi_2(s)dt$, $\Phi_1(s)=|s|^{\frac{1}{2}}\textrm{sgn}(s)$, $\Phi_2(s)=\Phi_1^{'}(s)\Phi_1(s)$.
    
    \item 
  the VGSTA (also denoted as LVGSTA due to the linear sliding surface)~\cite{gonzalez2012variable}, $u=\frac{1}{\bar{K}}\left( \ddot{r}-k_2\dot{e}+\bar{T}_{v}v \right)-\frac{h_1(\dot{v})}{\bar{K}}\Phi_1(s)-\int_0^{t}\frac{h_2(\dot{v})}{\bar{K}}\Phi_2(s)dt$, $\Phi_1(s)=|s|^{\frac{1}{2}}\textrm{sgn}(s)+h_3s$, $\Phi_2(s)=\Phi_1^{'}(s)\Phi_1(s)$.
    \item  the variable-gain PID (VGPID) controller used in ~\cite{li2015data}.
    \item the FCGSTA, $u=\frac{1}{\bar{K}}\left(\ddot{r}-k_1 \frac{dD^{\xi-1}\left[\textrm{sig}(e)\right]}{dt}-k_2\dot{e}+\bar{T}_{v}v \right)-\frac{h_1}{\bar{K}}\Phi_1(s)-\int_0^{t}\frac{h_2}{\bar{K}}\Phi_2(s)dt$, $\Phi_1(s)=|s|^{\frac{1}{2}}\textrm{sgn}(s)+h_3s$, $\Phi_2(s)=\Phi_1^{'}(s)\Phi_1(s)$.
     \item the proposed PFVSTA scheme (\ref{eq:u})-(\ref{equivalent control}).
    \item the proposed PFVSTA scheme with no feedforward term, denoted as incomplete PFVSTA (IFVSTA), $u=\frac{1}{\bar{K}}\left( -k_1 \frac{dD^{\xi-1}\left[\textrm{sig}(e)^a\right]}{dt}-\frac{k_2}{a}|e|^{\frac{1-a}{a}}\dot{e}+\bar{T}_{v}v \right)-\frac{h_1(\dot{v})}{\bar{K}}\Phi_1(s)-\int_0^{t}\frac{h_2(\dot{v})}{\bar{K}}\Phi_2(s)dt$, $\Phi_1(s)$ and $\Phi_2(s)$ are obtained from (\ref{phi1}) and (\ref{phi2}) respectively.
\end{itemize}

 

  \begin{figure}[http]
    \centering
     \includegraphics[width=230pt]{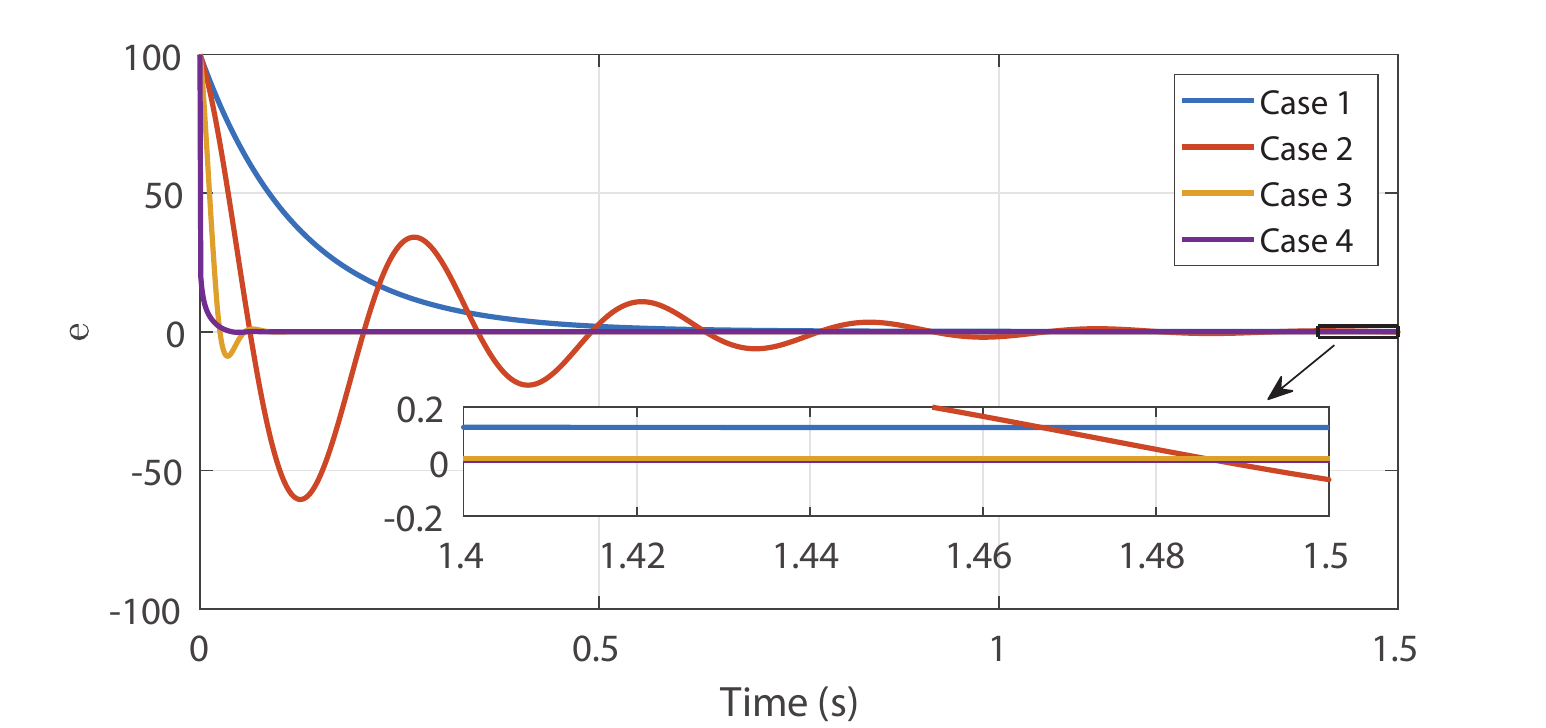}
     \caption{Dynamics of $e$ around different sliding surfaces.}

     \label{fig: sliding surface}
     \end{figure}

\subsection{Simulation Studies} \label{subsec:simulation results}

\subsubsection{Sliding Surface Comparison}
   \begin{table}[tp]  
    \centering
    \caption{Parameters of Controllers in Simulations.}  
    \label{tab: parameters of controller}
  \begin{tabular}{cccc}
    \cline{1-4} 
    Parameters & CGSTA & VGSTA & PFVSTA, IFVSTA \\
    \cline{1-4}
    $k_1$ & - & - & 53\\
    $k_2$ & 1200 & 1200 & 1200 \\
    $h_1$ & 1500  & $0.1 \dot{v}+50$ &$0.1 \dot{v}+50$\\
    $h_2$ & 10 & $0.1 \dot{v}+10$ & $0.1 \dot{v}+10$\\
    $h_3$ & - & 38 & -\\
\cline{1-4}
      \end{tabular}
    \end{table}
    \linespread{1}

Simulations are conducted first to investigate the proposed sliding surface (denoted as Case 4). Simulations are implemented in MATLAB/SIMULINK. For comparison, the simulations about the tracking error on the conventional linear sliding surface (LSS), the integral sliding surface (ISS), and the fractional-order sliding surface (FSS) are conducted. The three cases are denoted as Case 1, Case 2, and Case 3, respectively. According to the literature~\cite{gonzalez2012variable}, the LSS is formulated as
    $s=\dot{e}+k_2 e$
, the ISS is 
$
    s=\dot{e}+k_1 \int_0^t e+k_2 e
$
and the FSS is
$
   s= \dot{e}+k_1D^{\xi-1}\left[\textrm{sig}(e)\right]+k_2e
$.
Without loss of generality, parameters of the sliding surfaces are set as: $k_1=8$, $k_2=500$, $\xi=0.5$, $a=0.5$. To implement the fractional-order calculus in simulation, we use the discrete-time version to approximate the continuous fractional-order calculus~\cite{monje2010fractional}. 
In practice, $s$ cannot remain zero if uncertainties and disturbances exist in the system. Thus, $s$ is set as $s=1$ to investigate different sliding surfaces' performance with uncertainties. The simulation results are shown in Fig.~\ref{fig: sliding surface}.

From Fig.~\ref{fig: sliding surface}, it is apparent that the tracking error on the LSS converges to the equilibrium point at a prolonged rate. The tracking error on the ISS converges fast, but the trajectory is with a rather significant overshoot, and it takes a long time to converge. For the FSS, its trajectory converges faster than ISS, and it has a much smaller overshoot. As for our proposed sliding surface, tracking error converges the fastest, and it has a small overshoot. The enlarged window demonstrates the trajectories from the time $1.4$ s to $1.5$ s to represent how much the uncertainties influence the trajectories of each sliding surface. From the result, we notice that ISS's trajectory has not yet converged, LSS has a rather large distance to zero, and the trajectory of the proposed sliding surface is closest to zero. This phenomenon represents that the uncertainty on $s$ has the least influence on the state on the proposed sliding surface. 

\subsubsection{Controller Comparison in Acceleration Phases}
Based on the model in~Section~\ref{subsec: system_model}, we establish the simulation model of the wafer stage as
$
    G(s)=\frac{K}{s(s+T_v)}
$
, of which the input is the control signal and the output is the position of the mover. Parameters are chose as $T_v=1.092$ and $K=3.9124$, which were identified from the real system in our previous work~\cite{kuang2020fractional}.

\begin{figure}[http]
    \centering
     \includegraphics[width=230pt]{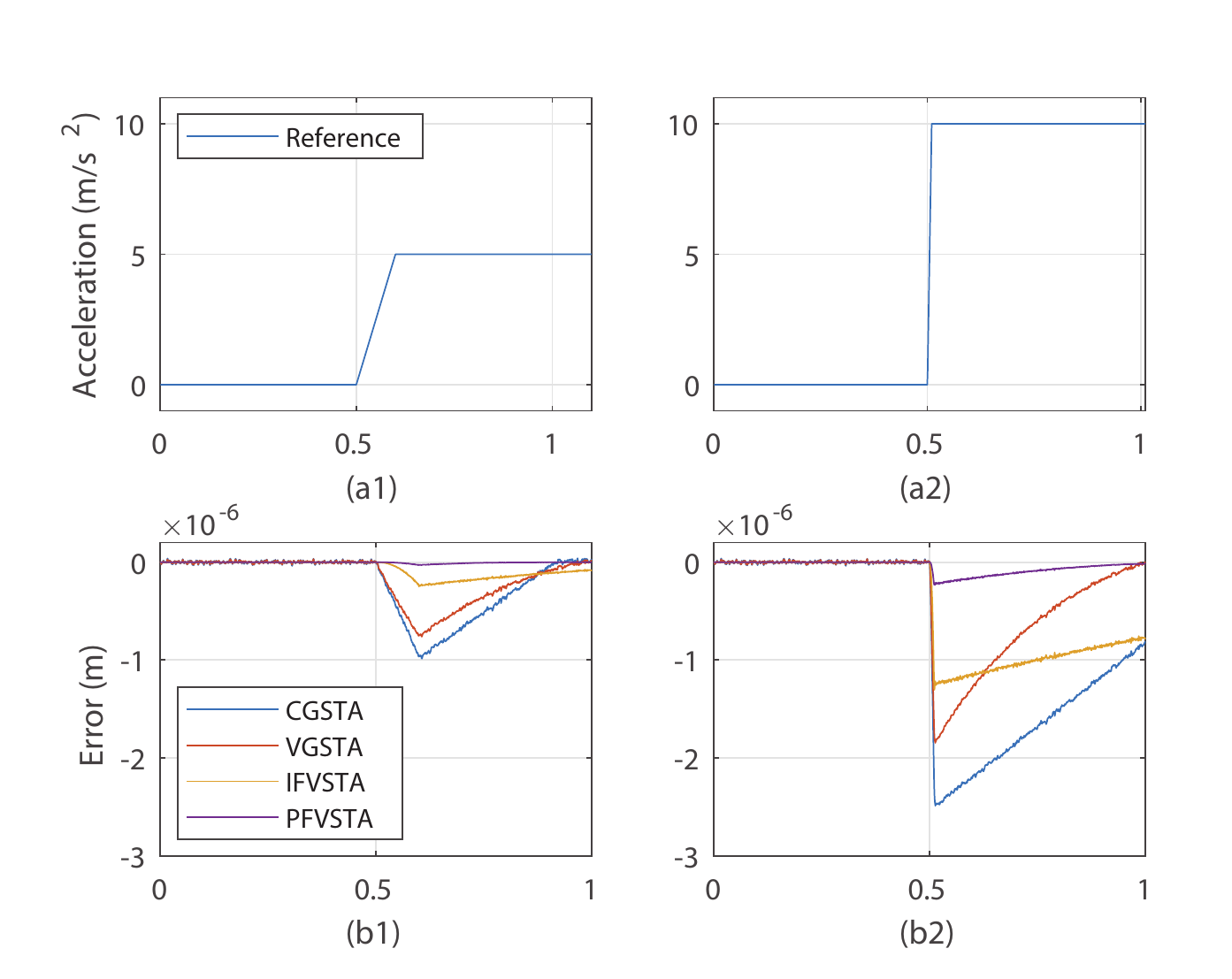}
     \caption{Performance of the controllers in the simulations of acceleration phases (the horizontal axis represents time, unit: s)}
     \label{fig:sim_error_and_control}
     \end{figure} 

The reference accelerations are designed as shown in Fig.~\ref{fig:sim_error_and_control}~(a1) and (a2). The reference position is obtained by integration. Note that the nominal parameters of the model are intended to be set as $\bar{T}_v=1$ and $\bar{K}=4$ to simulate the model uncertainties. Band-limited white noise is also introduced as measurement noises of the system. We compare four control algorithms in this group of simulations: CGSTA, VGSTA, IFVSTA, and PFVSTA. 
Their parameters are shown in Table~\ref{tab: parameters of controller} and the results are shown in Fig.~\ref{fig:sim_error_and_control}.

\begin{figure}[http]
      \centering
       \includegraphics[width=230pt]{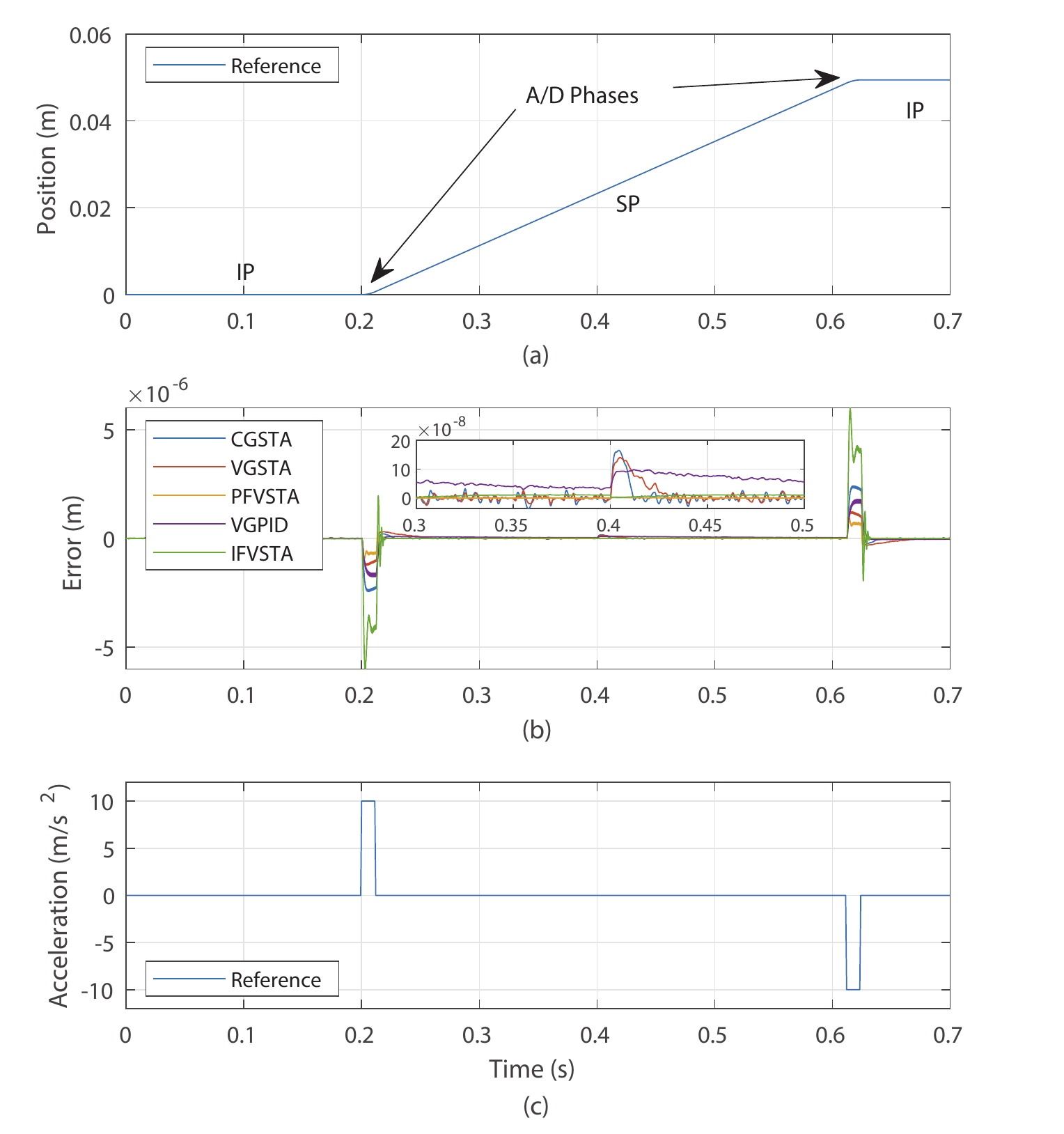}
       \caption{Performance of the controllers in simulational scanning phases. }
       \label{fig:sim_scan}
       \end{figure}

By comparing Fig.~\ref{fig:sim_error_and_control}~(b1)  with (a1) and (b2) with (a2), we can see that tracking errors of all the four controllers increase with the acceleration. This phenomenon accords well with our theoretical analysis in Remark~\ref{remark1}. We can see that a larger acceleration results in a larger uncertainty to the system from~(\ref{ineq:boundary of dg2}), which will deteriorate the system's precision. Results shown in Fig.~\ref{fig:sim_error_and_control} can also reflect the necessity to investigate how the acceleration of reference influences the performance. In wafer stage systems, the acceleration phases usually last less than 0.1 s, and the scanning process begins just after the acceleration. Fig.~\ref{fig:sim_error_and_control} demonstrates that the errors of most controllers cannot converge to zero in 0.1 s. Therefore the performance in the acceleration phases will influence the scanning process as well.


Besides, from the results, PFVSTA's error trajectory is always the smallest throughout the simulations, which shows that our method has a significant advantage over the other methods. Moreover, the number of tunable parameters in our proposed method is the same as other methods. It means that the improvement in precision will not bring difficulties to tune parameters in the controller. Moreover, by comparing PFVSTA and IFVSTA in Fig.~\ref{fig:sim_error_and_control}~(b1) and (b2) and considering the difference between PFVSTA and IFVSTA, we can conclude that the feedforward term in proposed $u_{eq}$ is necessary to get the excellent performance. This conclusion will also be proven more than one time in the following simulations and experiments.


     
 \subsubsection{Controller Comparison in Scanning Process}    
The second group of simulations is conducted to simulate the scanning process. As shown in Fig.~\ref{fig:sim_scan}~(a), a typical scanning trajectory contains at least three phases: the idle phase (IP), the scanning phase (SP), and the acceleration/deceleration phase (A/D phase). The scan length is set as 0.05 m, the first idle phase is set as 0.2 s, the acceleration and deceleration phases are set as 0.012 s, the scanning phase is set as 0.4 s, and the maximum acceleration is set as 10 m/s$^2$. Particularly, when the moving stage moves 0.23 m (around 0.4 s in this case), an extra step disturbance with an amplitude of 0.1 is applied intentionally to investigate the robustness of controllers. 

Apart from the four control algorithms in the previous group, we also implemented the VGPID controller as a comparison to study the robustness of the proposed PFVSTA.
Parameter of PFVSTA are tuned as $k_1=33$, $k_2=10000$, $h_1=0.01\dot{v}+40$, and $h_2=0.01\dot{v}+20$. Parameters of VGSTA are $k_2=1200$, $h_1=0.5\dot{v}+50$, $h_2=0.1 \dot{v}+20$ and $h_3=43$. CGSTA has parameters of $k_2=1200$, $h_1=1500$, $h_2=80$. As for VGPID, the gains for the proportional term, integral term and derivative term are $K_{p}=$1.2$\times$10$^6$, $K_{i}=$8$\times$10$^6$ and $K_{d}=$3$\times$10$^3$. The variation of the proportion are $\Delta K_{p1}=$0.5$\times$10$^6$ in the A/D phases and $\Delta K_{p2}=$0.2$\times$10$^6$ in other phases.

The results are shown in Fig.~\ref{fig:sim_scan}~(b). We can see that in the A/D phases, CGSTA has the largest error, which proves that compared with variable-gain methods, constant gain methods cannot get satisfying performance. This proves further that the variable-gain structure is helpful to reduce the influence of reference acceleration and indirectly prove the rationality of the variable structure in our proposed $u_{sw}$~(\ref{eq:usw}). From the enlarged window, we can see that the proposed PFVSTA is smooth and precise in the scanning phase, while the VGPID is not as precise as other methods. Moreover, after the extra disturbance is applied, position errors of CGSTA and VGSTA converge quickly, while VGPID converges rather slow. This shows that compared with the SMC method, the PID method is easy to be influenced by extra disturbances. As for the proposed PFVSTA method, the position error is impervious to the disturbance, which shows its excellent robustness.

\subsection{Experimental Setup} \label{subsec:experimental setup}

\begin{figure}[http]
  \centering
   \includegraphics[width=230pt]{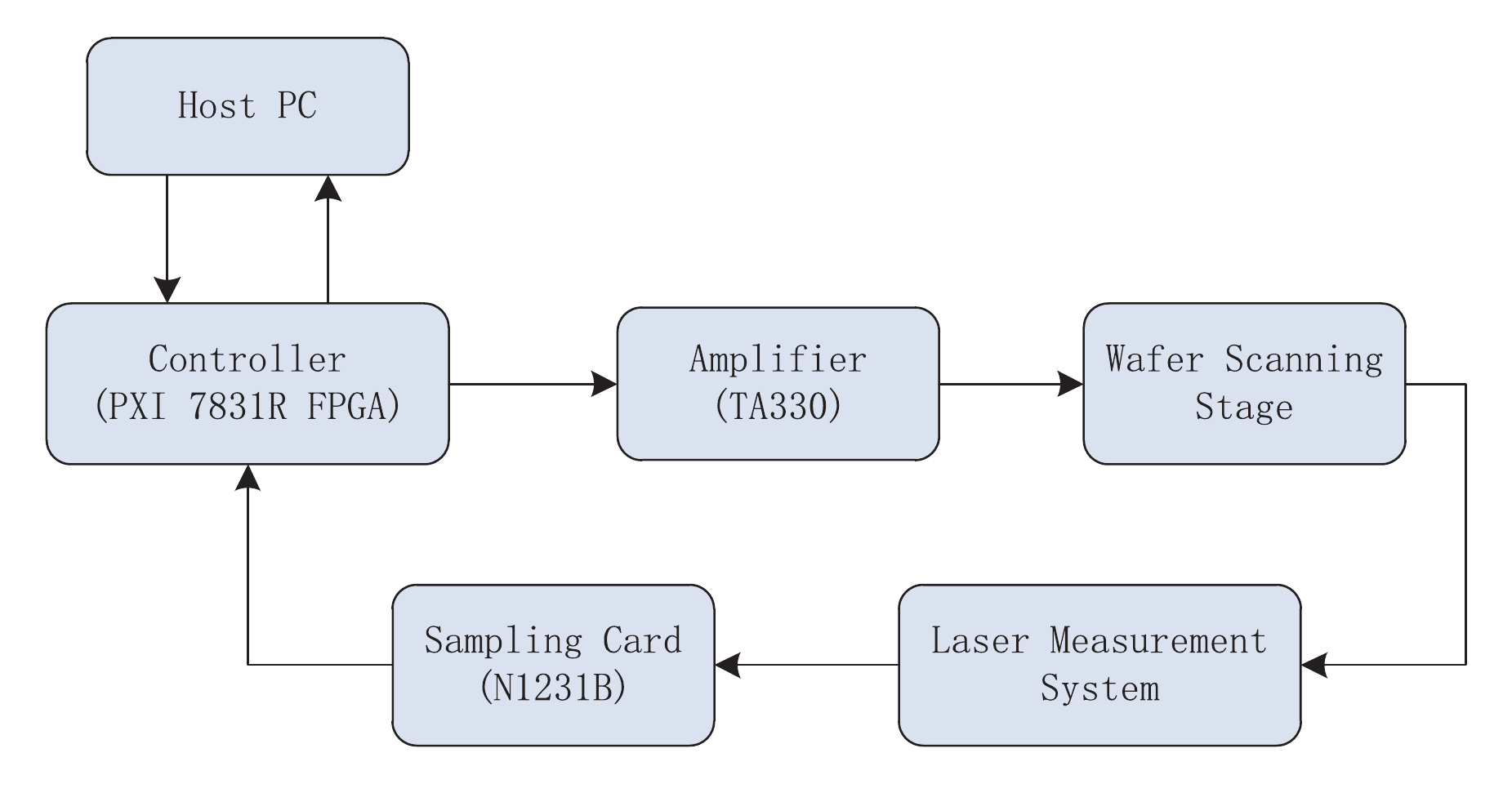}
   \caption{The structure of the experimental system.}
   \label{fig:structure_closed_loop}
   \end{figure}

   \begin{figure}[http]
  \centering
   \includegraphics[width=230pt]{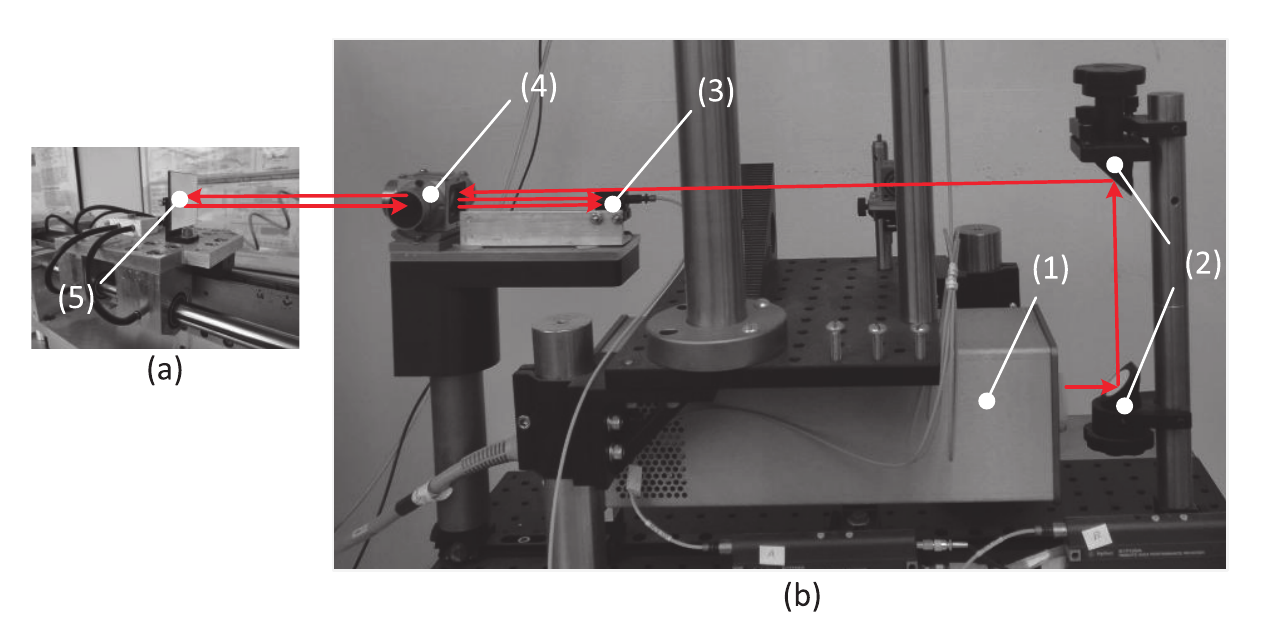}
   \caption{The laser system for the measurement of the wafer stage position, including (1)~a laser generator, (2)~two reflectors, (3)~an optical pickup, (4)~an interferometer, (5)~a reflecting mirror. (a) is part of wafer stage and (b) is part of laser measurement system. The red arrows stand for the optical path.}
   \label{fig:laser}
   \end{figure}
   
We also tested the proposed algorithm on the wafer stage shown in Fig.~\ref{fig:wafer_stage}. As we mentioned before, the testbed mainly consists of a linear motor and a moving stage. The moving stage is driven by the linear motor along with the sliding guide. Between the sliding guide and the moving stage, air bearings are fixed to reduce the friction. As shown in Fig.~\ref{fig:structure_closed_loop}, the experimental system includes a host computer, a controller, an amplifier, the wafer stage, a laser measurement system, and a sampling card. Control algorithms are implemented with the LabView Real-Time system running on the controller (a field-programmable gate array, i.e., FPGA, PXI 7831R from National Instruments). The control signals for the stages and counter-masses, generated by the FPGA controllers, are amplified via separate amplifiers (TA330 from Trust Automation company) and then sent to the linear motor of the wafer stage. In the meanwhile, the position signal from the laser measurement system (in~Fig.~\ref{fig:laser}) is sampled by a sampling card (N1231B from Keysight) and fed back into the controller running on the FPGA. 

In this paper, two groups of experiments are conducted, which are denoted as Case~1 and Case~2. In both cases, the scan length is set as 0.05 m, the scan velocity is set as 0.1 m/s, the waiting time before the scan is set as 0.2 s, and the holding time at the top of the scan is set as 0.1 s. The difference is that the maximum acceleration is set as 1.5 $m/s^2$ in Case 1 and 10 $m/s^2$ in Case 2. The reference scanning trajectories are displayed in Fig.~\ref{fig:ex_case1}~(a) and Fig.~\ref{fig:ex_case2}~(a) respectively.

In either case, four controllers are compared: LVGSTA, FCGSTA, IFVSTA and PFVSTA. For the proposed PFVSTA and IFVSTA, $\xi$ and $a$ are directly chosen as $1/2$. Other parameters are tuned to get the best performance, and they are $k_1=100$, $k_2=2\times 10^7$, $h_1(\dot{v})=550|\dot{v}|+13$ and $h_2(\dot{v})=10|\dot{v}|+4$. For the other two control schemes, to make the comparison fair, all the parameters are also tuned to guarantee that the best performance is attained.


   \subsection{Experimental Results} \label{subsec:experimental results}
   
   \begin{figure}[http]
    \centering
     \includegraphics[width=230pt]{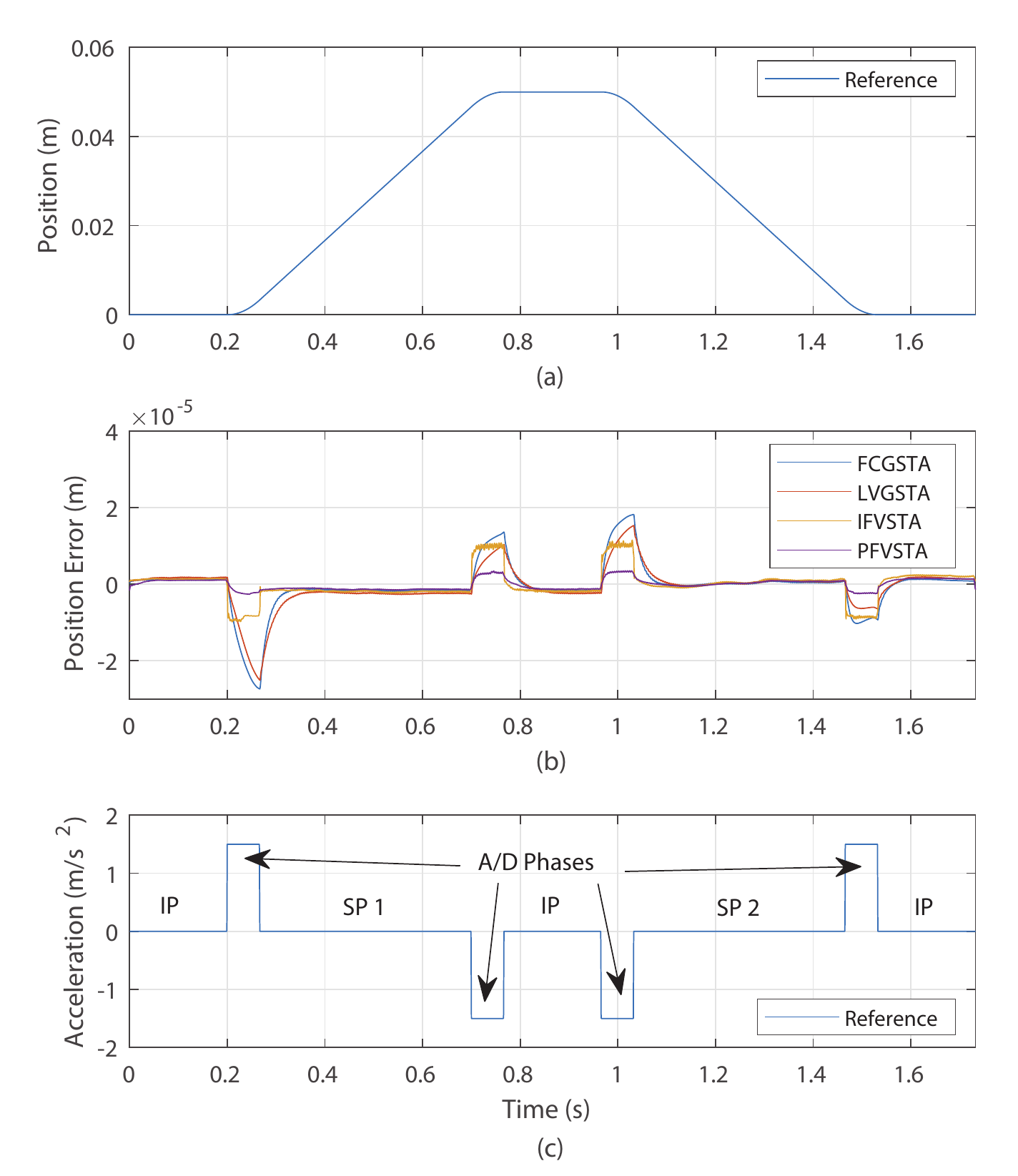}
     \caption{Experimental results of the controllers in Case~1.} 
     \label{fig:ex_case1}
     \end{figure}
     
Fig.~\ref{fig:ex_case1} depicts the experimental results in Case 1. 
From Fig.~\ref{fig:ex_case1} (b), we can see that all the four controllers have relatively larger absolute values of position errors during the A/D phases while having relatively smaller absolute values during the idle phases and scanning phases.
We can also see that either in idle phases, scanning phases, or A/D phases, the proposed PFVSTA has the smallest error overall. For the FCGSTA, we can find that it usually has the largest error in the A/D phases, while it has a medium error in the idle and scanning phases. This phenomenon proves that the reference's acceleration is more likely to influence the performance of FCGSTA than LVGSTA and PFVSTA. Considering the differences among these controllers, it reflects the variable-gain structure, the designed sliding surface~(\ref{def:sliding surface}) and the designed switching control law~(\ref{eq:usw}) help to improve the performance in A/D phases. 


Fig.~\ref{fig:ex_case2} shows the results of the experiments in Case 2. As the reference trajectory's maximum acceleration is much larger than the experiments in Case 1, the corners of the reference trajectory are apparently sharper. 
For the position errors, the peak values of FCGSTA and LVGSTA reach at about 100 $\mu m$, while the peak values of IFVSTA are the largest. The error of the proposed PFVSTA still keeps the smallest value in all three different phases. This shows that the PFVSTA remains its high precision even when the reference acceleration is raised (compared with Case 1). 
From the enlarged window in Fig.~\ref{fig:ex_case2}~(b), we can also observe that in the scanning phase, PFVSTA is the smoothest in the four controllers, which proves that the chattering phenomenon is significantly reduced.
All of these phenomena are consistent with the previous simulation results and theoretical analysis. 

     \begin{figure}[http]
      \centering
       \includegraphics[width=230pt]{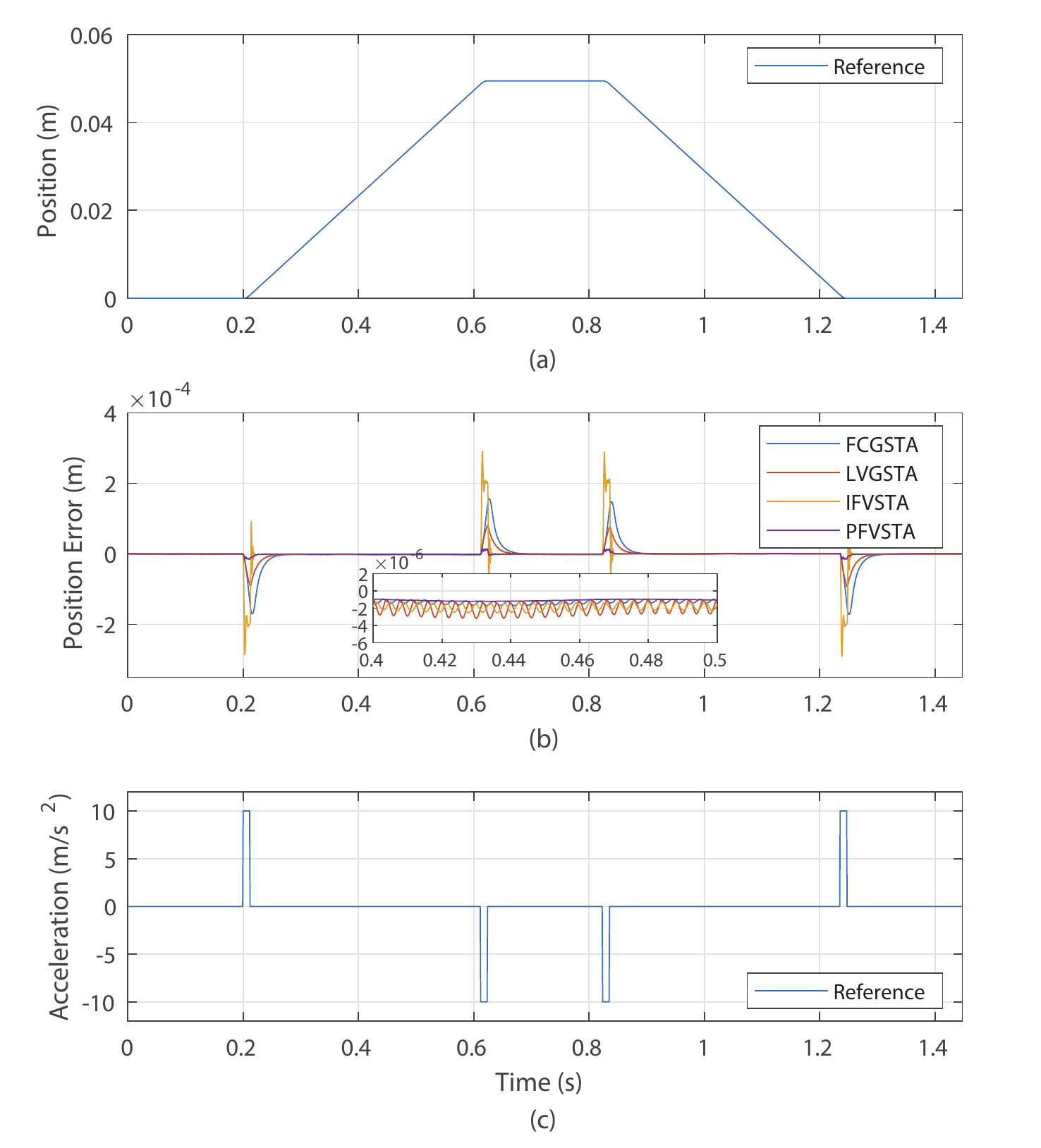}
       \caption{Experimental results of the controllers in Case~2.}
       \label{fig:ex_case2}
       \end{figure}

\subsection{Quantitative Analysis}

     \begin{figure}[http]
      \centering
       \includegraphics[width=230pt]{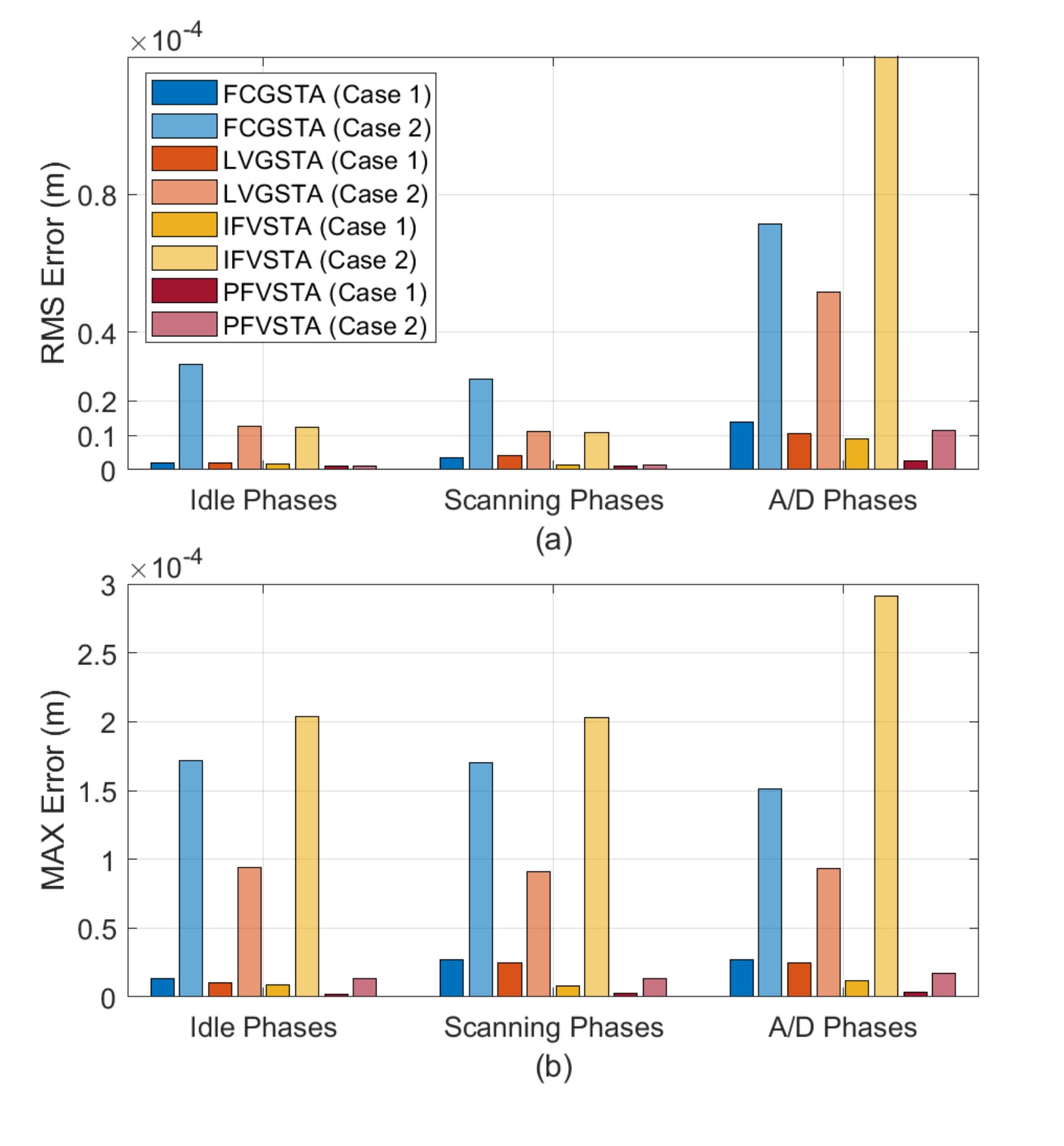}
       \caption{Statistical characteristics of different controllers under the two cases: (a) the RMS Errors, and (b) the MAX Errors. The RMS Error of IFVSTA in Case 2 is $1.964 \times 10^{-4}$ m}
       \label{fig:ex_error_analysis}
       \end{figure}
       

Quantitative analysis of the experiments is conducted to investigate the results thoroughly. Because the root-mean-square (RMS) error and the maximum (MAX) error are two critical indicators to describe the precision, we figure them out in each kind of phases and display the bar charts in Fig.~\ref{fig:ex_error_analysis}. The RMS error is calculated based on
$
    RMS~Error=\sqrt{\frac{\sum_{i=1}^{n}e(i)}{n}}
$
and the MAX error is based on
$
    MAX~Error=\max {|e(i)|}, i=1,2,...,n,
$
where $e(i)$ is the sampled position error, $i$ denotes the serial number and $n$ is the total number of the sampled errors.

Fig.~\ref{fig:ex_error_analysis} (a) shows that of all the three phases, the PFVSTA holds the smallest RMS error in both Case 1 and Case 2. Comparing FCGSTA with LVGSTA, we can see that when the acceleration is small in Case 1, differences between errors of FCGSTA and LVGSTA are not as much as those when the acceleration is large in Case 2. We can attribute the reason for the effects of the variable gains. Besides, for the scanning phases, which we care about the most, the proposed method's RMS error changes much less between Case 1 and Case 2 than the other methods. Due to the combination of the variable gain method and FSS, the performance of the proposed PFVSTA tends to be less sensitive to the accelerations in the reference signal. The contrast between IFVSTA and PFVSTA proves that our proposed method's feedforward term is essential to reduce the tracking error.

All the maximum errors are depicted in Fig.~\ref{fig:ex_error_analysis}~(b). Not surprisingly, the PFVSTA still keeps the best precision according to the statistics. An interesting phenomenon is that the controllers' maximum errors in scanning phases are similar to those in the A/D phases. Checking back the error trajectories in Fig.~\ref{fig:ex_case1}~(b) and Fig.~\ref{fig:ex_case2} (b), we can figure out the reason is that the maximum errors always happen at the junctions of different phases. In this way, the bad performance in the A/D phases negatively influences the performance in the scanning phases. This proves that investigating the performance in A/D phases is also important to improve the precision of wafer stages. 

\begin{figure}[http]
      \centering
       \includegraphics[width=230pt]{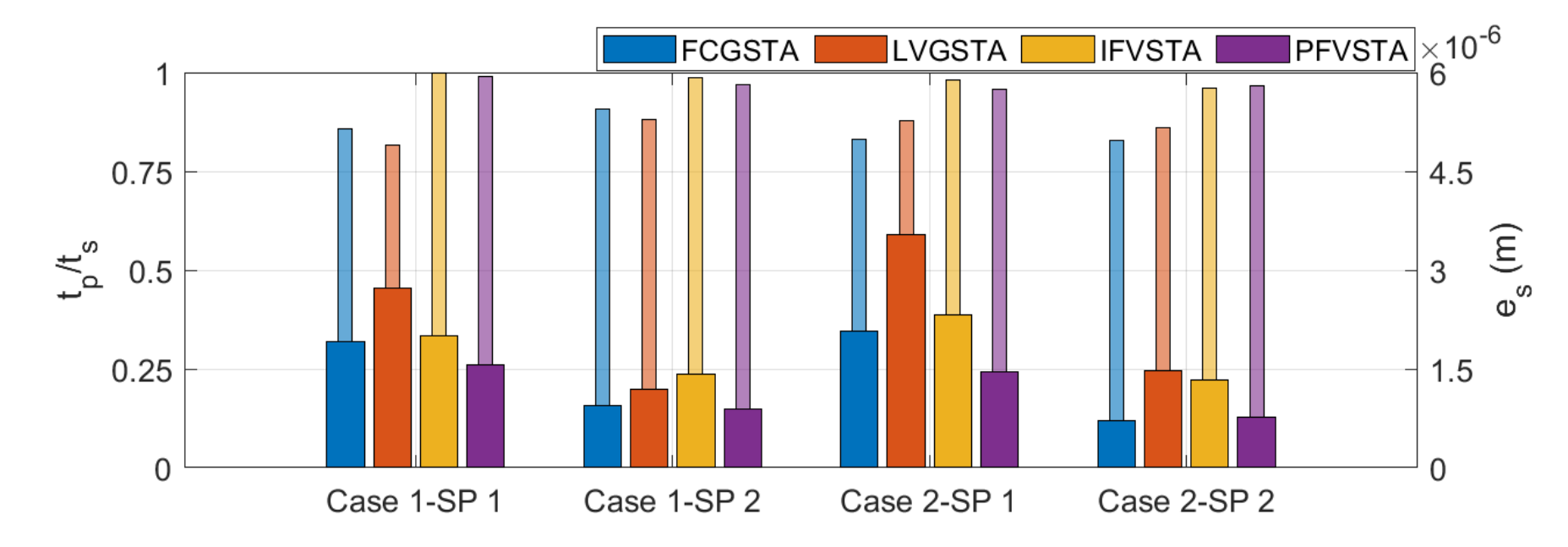}
       \caption{The ratios between the valid scanning time and the whole scanning time (shown as narrow bars), and the valid scanning errors (shown as wide bars) of different controllers.}
       \label{fig:error_ana_scanning}
       \end{figure}

Considering that the precision of the scanning phases is the most significant in the industrial application and that there is an obvious decline of the position error when the scanning phases begin, a possible approach is to let the machine begin to scan after the error drops to a satisfying value to guarantee the precision.
We define the time of the whole scanning phase as $t_s$. After the position error $e$ enters into a certain range, it never escapes from it. Then half of the minimum width of the range is defined as the valid scanning error $e_s$, and the time it keeps in the range is defined as valid scanning time $t_p$. Based on these definitions, it is apparent that a smaller value of $e_s$ and a larger value of $t_p/t_s$ means better performance with higher precision and efficiency. The statistics of $e_s$ and $t_p/t_s$ are demonstrated as bar charts in Fig.~\ref{fig:error_ana_scanning}. It is clear that in the four scanning phases of two cases, both IFVSTA and PFVSTA have the largest values of $t_p/t_s$, which is very close to 1. Furthermore, PFVSTA has the smallest $e_s$ under almost all the situations. This result proves the overall advantage of the proposed wafer stage control method once more.



\section{Conclusion}  \label{sec:conclusion}
In this paper, a fractional-order variable-gain super-twisting controller was proposed for the wafer stage system, with a novel sliding surface and a novel variable-gain structure of the super-twisting algorithm. Specifically,  we considered the practical situation during the analysis of stability and analyzed the controller's theoretical precision. These theoretical analyses explained the advantages of the proposed method. Moreover, simulation and experimental results of the sliding surfaces and the whole control algorithm not only correlated well with the theoretical analyses but also proved the proposed method had at least three merits: the high precision, the robustness to uncertainties, and the small chattering amplitude, which guarantee the excellent performance on wafer stages. Finally, adequate experimental results and their quantitative analyses proved the superiority of the proposed algorithm with application to wafer stages.

\appendices
\section{}
\begin{mydef} (see \cite{podlubny1998fractional}) \label{def: fractional-order calculus}
  The Riemann-Liouville definition of the $\xi$th order derivative and integration for function $f(t)$ are respectively defined as
  \begin{align}
  D^{\xi}f(t)=\frac{1}{\Gamma(m-\xi)}\frac{d^m}{dt^m}\int^t_{t_0}\frac{f(\tau)}{(t-\tau)^{\xi-m+1}}d\tau
  \end{align}
  \begin{align}
    {}_{t_0}I^{\xi}_{t}f(t)=\frac{1}{\Gamma(\xi)}\int^t_{t_0}\frac{f(\tau)}{(t-\tau)^{1-\xi}}d\tau
  \end{align}
  where $m=\lceil \xi \rceil$, $\Gamma(\bullet)$ is the Gamma function, and it is computed as $\Gamma(x)=\int_{0}^{\infty}e^{-t}t^{x-1}dt$. 
\end{mydef}




\ifCLASSOPTIONcaptionsoff
  \newpage
\fi



\bibliographystyle{IEEEtran}

\bibliography{paper10}
\end{document}